\documentclass[%
 aps,
 prx,
 reprint,
 amsmath,
 amssymb,
 amsthm,
 showkeys,
 superscriptaddress,
 eprint,
 nofootinbib,
 bibnotes,
 twoside
]{revtex4-2}

\usepackage[utf8]{inputenc}
\usepackage[T1]{fontenc}
\usepackage{graphicx}
\usepackage{xcolor}
\usepackage{enumerate}
\usepackage{amsfonts,amsmath,amssymb,amsthm}
\usepackage{dsfont}
\usepackage{mathrsfs}
\usepackage{bm}

\usepackage{hyperref}
\hypersetup{
	colorlinks = true,
	linkcolor  = blue,
	citecolor  = blue,
	urlcolor   = blue,
	filecolor  = black,
	linktoc    = all,
}
\usepackage[numbered]{bookmark}

\newcommand{\one}{\mathds{1}}

\newcommand{\rr}{ {\hspace{1pt} \| \hspace{1pt}} }

\newcommand{\MM}{\mathcal{M}}

\newcommand{\HH}{\mathcal{H}}
\newcommand{\NN}{\mathcal{N}}

\newcommand{\EE}{\mathcal{E}}

\newcommand{\LL}{\mathcal{L}}
\newcommand{\KK}{\mathcal{K}}

\newcommand{\Hmin}{\widehat{H}}

\newcommand{\ie}{{\textit{i.e.}}}

\newcommand{\ket}[1]{|{#1}\rangle}
\newcommand{\bra}[1]{\langle{#1}|}
\newcommand{\ketbra}[1]{{\ket{#1} \! \bra{#1}}}

\newtheorem{thm}{Theorem}

\begin{document}

\title{Non-Multiplicativity of the Holevo Barycenter of Quantum Channels}

\author{Sayantan Chakraborty}
\affiliation{Department Mathematik/Informatik---Abteilung Informatik, Universit\"at zu K\"oln, Albertus-Magnus-Platz, D-50923 K\"oln, Germany}

\author{Stefano Mancini}
\affiliation{School of Science and Technology, University of Camerino, Via Madonna delle Carceri 9, I-62032 Camerino, Italy}
\affiliation{Istituto Nazionale di Fisica Nucleare, Sezione di Perugia, Via A. Pascoli, I-06123 Perugia, Italy}

\author{Leonardo Rossetti}
\affiliation{Department Mathematik/Informatik---Abteilung Informatik, Universit\"at zu K\"oln, Albertus-Magnus-Platz, D-50923 K\"oln, Germany}
\affiliation{School of Science and Technology, University of Camerino, Via Madonna delle Carceri 9, I-62032 Camerino, Italy}
\affiliation{Istituto Nazionale di Fisica Nucleare, Sezione di Perugia, Via A. Pascoli, I-06123 Perugia, Italy}

\author{Andreas Winter}
\affiliation{Department Mathematik/Informatik---Abteilung Informatik, Universit\"at zu K\"oln, Albertus-Magnus-Platz, D-50923 K\"oln, Germany}
\affiliation{ICREA \&{} F\'isica Te\`orica: Informaci\'o i Fen\`omens Qu\`antics, Departament de F\'isica, Universitat Aut\`onoma de Barcelona, ES-08193 Bellaterra, Spain}
%\affiliation{ICREA, Passeig Llu\'is Companys 23, ES-08010 Barcelona, Spain}
%\affiliation{Institute for Advanced Study, Technische Universit\"at M\"unchen, Lichtenbergstra{\ss}e 2a, D-85748 Garching, Germany}

\begin{abstract}
The Holevo barycenter of a quantum channel is the unique output state obtained as the average output of any ensemble achieving the Holevo capacity. Given two quantum channels, the multiplicativity problem asks whether this barycenter tensorizes under parallel composition, namely whether the barycenter of the product channel coincides with the tensor product of the individual barycenters. This question is closely related to the additivity problem for the Holevo capacity: additivity implies tensorization of the Holevo barycenter, while tensorization alone is not sufficient for additivity.
Although a construction is known that demonstrates the existence of channels with non-additive Holevo capacity, their corresponding Holevo barycenters still tensorize, leaving open whether multiplicativity might ultimately hold universally.
%Although the Shor--Hastings construction guarantees the existence of channels with non-additive Holevo capacity, the corresponding Holevo barycenters still tensorize, leaving open whether multiplicativity might ultimately hold universally.
Here, we answer this question in the negative by exhibiting channels for which the Holevo barycenter is not multiplicative under the tensor product of the channel with itself. Moreover, we show that the entropy of the Holevo barycenter is neither universally subadditive nor universally superadditive under tensor product.
\end{abstract}

\date{8 September 2026}

\maketitle

%\aw{Are we aiming for PRL? Then section headings should be smaller and starting a paragraph, rather than the space-wasting titles -- this would give a clearer picture of the space we have...}

\section{Introduction}
\label{sec:intro}
The classical communication capacity of a quantum channel quantifies the asymptotic rate at which classical information can be reliably transmitted. The evaluation of this quantity is, in general, a regularized problem, since one must in principle optimize over encodings involving arbitrarily many uses of the channel. The Holevo capacity is the corresponding single-letter quantity and, by the Holevo--Schumacher--Westmoreland theorem, gives the asymptotic communication rate when codewords are restricted to product-state encodings across channel uses~\cite{Schumacher_Westmoreland_1997,Holevo_1998}. If the Holevo capacity were additive under tensor products, then the regularization would no longer be needed, and the classical capacity would reduce to the Holevo capacity. The additivity of the Holevo capacity was a central open problem in quantum information theory for many years~\cite{Amosov_Holevo_Werner_2000}. It was ultimately settled in the negative by Hastings' counterexample to the additivity of the minimum output entropy~\cite{Hastings_2009}, after Shor had established the equivalence between the corresponding additivity conjectures for the Holevo capacity and the minimum output entropy~\cite{Shor_2004}.

Although the optimization defining the Holevo capacity may admit several capacity-achieving ensembles, their average output state is unique~\cite{schumacher1999optimalsignalensembles,Shirokov_2006_OptimalSets}. 
Throughout this work, we refer to this state as the \textit{Holevo barycenter}. In the literature, the same state is also referred to as the \textit{output optimal average state}~\cite{Shirokov_2006_HolevoCapacity}. The multiplicativity problem for the Holevo barycenter asks whether it tensorizes under parallel composition. Previous work by Shirokov establishes a direct connection with the additivity problem for the Holevo capacity: additivity forces the barycenter to tensorize, whereas multiplicativity alone is not sufficient for additivity~\cite{Shirokov_2006_OptimalSets}. Thus, non-additivity of the capacity is necessary, but may not be sufficient, for non-multiplicativity of the barycenter.

This logical asymmetry is already visible in the Shor–Hastings route~\cite{Shor_2004,Hastings_2009}. Shor’s reduction technique, applied to Hastings’ counterexamples, produces channels with non-additive Holevo capacity, but the twirling employed by the construction yields capacity-achieving ensembles with maximally mixed average outputs, both for the individual enlarged channels and for their product. The corresponding Holevo barycenters therefore still tensorize. At this stage, one might still suspect that the Holevo barycenter tensorizes universally.

In the present work, we rule out this possibility by exhibiting channels which, when tensored with itself violate tensor multiplicativity of the Holevo barycenter. Our construction relies on direct sums of quantum channels~\cite{Fukuda_2007}. For a direct-sum channel, the barycenter is block diagonal and its weights are determined by the capacities of the component channels; under tensor product, they are instead controlled by the capacities of the pairwise product channels. Starting from a pair with strictly superadditive Holevo capacity and adjoining an entanglement-breaking channel, whose Holevo capacity is additive with that of every channel~\cite{Shor_2002}, we obtain incompatible block weights and hence a non-multiplicative Holevo barycenter. Finally, by specializing the construction to two components and choosing a suitable depolarizing channel, we show that there is no universal entropic ordering: the entropy of the barycenter of the tensor square can be either smaller or larger than twice the entropy of the individual barycenter. Hence, the entropy of the Holevo barycenter is neither universally subadditive nor universally superadditive under tensorization.

\section{The Holevo Barycenter and the Multiplicativity Problem}
\label{Sec 2}
Throughout, all Hilbert spaces are assumed to be finite-dimensional, and
$\LL(\mathcal H)$ denotes the space of linear operators on a Hilbert space
$\mathcal H$. 

Let $\NN:\LL(\HH)\to\LL(\KK)$ be a quantum channel with input space
$\HH$ and output space $\KK$. For an input ensemble $\EE=\{p_i,\rho_i\}_i$, let $\Omega(\EE)\equiv\sum_i p_i\NN(\rho_i)$ denote its average output state. The Holevo capacity $\chi(\NN)$ admits the equivalent expressions~\cite{holevo1973,schumacher1999optimalsignalensembles}
\begin{equation}
\begin{aligned}
    \chi(\NN)
    &\equiv\max_{\EE}
    \Bigl[H(\Omega(\EE))-\sum_i p_iH(\NN(\rho_i))\Bigr]
    \\
    &=\max_{\EE}\sum_i p_i
    D\bigl(\NN(\rho_i)\rr\Omega(\EE)\bigr),
\end{aligned}
\label{eq: chi N}
\end{equation}
where $H$ denotes the
von Neumann entropy and $D$ the quantum relative entropy.

Let $\EE_1$ and $\EE_2$ be optimal ensembles for~\eqref{eq: chi N}, and let $\bar{\EE}$ be their equally weighted mixture. By concavity of the von Neumann entropy, $\bar{\EE}$ is also optimal. Applying Donald's identity~\cite{schumacher1999optimalsignalensembles} separately to the two components of $\bar{\EE}$ we obtain
\begin{equation*}
    \chi(\NN)=\chi(\NN)+\frac{1}{2}D(\Omega(\EE_1)\rr\Omega(\bar{\EE}))+\frac{1}{2}D(\Omega(\EE_2)\rr\Omega(\bar{\EE})).
\end{equation*}
Therefore, all optimal ensembles have the same average output state. We denote this unique state by $\Omega(\NN)$ and refer to it as the Holevo barycenter of $\NN$.

Let $\MM_1$ and $\MM_2$ be quantum channels for which the Holevo capacity is additive, \ie
\begin{equation}
\label{eq: add Hol cap}
    \chi(\MM_1\otimes\MM_2)=\chi(\MM_1)+\chi(\MM_2).
\end{equation}
Then the Holevo barycenter tensorizes:
\begin{equation*}
    \Omega(\MM_1\otimes\MM_2)
    =\Omega(\MM_1)\otimes\Omega(\MM_2).
\end{equation*}
Indeed, whenever Eq.~\eqref{eq: add Hol cap} holds, the tensor product of any two optimal ensembles for $\MM_1$ and $\MM_2$, respectively, is optimal for the product channel. Its average output state is $\Omega(\MM_1)\otimes\Omega(\MM_2)$, and the claim follows from the uniqueness of the Holevo barycenter. The converse implication, however, fails in general. Using a convex-analytic study of the associated optimal sets, Shirokov obtained a necessary-and-sufficient characterization of additivity: multiplicativity of the Holevo barycenter must be supplemented by an additional geometric condition on the optimal set of the product channel~\cite{Shirokov_2006_OptimalSets}. The Shor-Hastings construction exhibits this asymmetry explicitly \cite{Shor_2004,Hastings_2009}.

The minimum output entropy of a channel $\NN$ is defined as 
\begin{equation*}
\label{eq: min out ent}
  \Hmin(\NN) \equiv \min_{\ket{\psi}} H\bigl(\NN(\ketbra{\psi})\bigr).
\end{equation*}
If the output space $\KK$ has dimension $d$, then
\begin{equation}
\label{eq: bound for chi}
    \chi(\NN)\leq\log d-\Hmin(\NN).
\end{equation}
Let $\NN_1:\LL(\HH_1)\to\LL(\KK_1)$ and $\NN_2:\LL(\HH_2)\to\LL(\KK_2)$ be quantum channels whose output spaces have the same dimension, $\dim\KK_1=\dim\KK_2=d$, and for which the minimum output entropy is strictly subadditive, \ie
\begin{equation*}
\Hmin(\NN_1\otimes\NN_2)
<
\Hmin(\NN_1)+\Hmin(\NN_2).
\end{equation*}
The existence of such pairs follows from Hastings' result~\cite{Hastings_2009}. For $\ell=1,2$, let $\mathcal X_\ell$ be a $d^2$-dimensional Hilbert space with an orthonormal basis
$\{\ket{x_\ell}\}_{x_\ell=0}^{d^2-1}$, and define the enlarged channel
$\NN_\ell':\LL(\HH_\ell\otimes\mathcal X_\ell)\to\LL(\KK_\ell)$ by
\begin{equation*}
    \NN_\ell'(\rho)
    =
    \sum_{x_\ell=0}^{d^2-1}
    U_{x_\ell}^{(\ell)}
    \NN_\ell\bigl(
        \bra{x_\ell}\rho\ket{x_\ell}
    \bigr)
    U_{x_\ell}^{(\ell)\dagger},
\end{equation*}
where, for each $\ell=1,2$,
$\{U_x^{(\ell)}\}_{x=0}^{d^2-1}$ is a perfectly randomizing set of
unitaries on $\KK_\ell$.
Let $\ketbra{v_\ell}$ attain $\Hmin(\NN_\ell)$, for $\ell=1,2$, and let $\ketbra{w}$ attain $\Hmin(\NN_1\otimes\NN_2)$. By concavity and unitary invariance of the von Neumann entropy,
$\Hmin(\NN_\ell')=\Hmin(\NN_\ell)$ and
$\Hmin(\NN_1'\otimes\NN_2')=\Hmin(\NN_1\otimes\NN_2)$.
Consider the ensembles
$\{\ketbra{v_\ell}\otimes\ketbra{x_\ell}\}_{x_\ell}$, for $\ell=1,2$, and
$\{\ketbra{w}\otimes\ketbra{x_1}\otimes\ketbra{x_2}\}_{x_1,x_2}$,
all taken with uniform probabilities. Their average outputs are maximally mixed, while every signal state produces an output with the corresponding minimum entropy. They therefore saturate the bound in~\eqref{eq: bound for chi} and are optimal for $\NN_\ell'$ and $\NN_1'\otimes\NN_2'$, respectively. Consequently,
\begin{equation*}
    \chi(\NN_1'\otimes\NN_2')
    >\chi(\NN_1')+\chi(\NN_2').
\end{equation*}
This is precisely Shor's reduction from non-additivity of the minimum output entropy to non-additivity of the Holevo capacity~\cite{Shor_2004}. Since the optimal ensembles above have maximally mixed average outputs, uniqueness of the Holevo barycenter also gives
\begin{equation*}
    \Omega(\NN_1'\otimes\NN_2')
    =\Omega(\NN_1')\otimes\Omega(\NN_2').
\end{equation*}
Thus, the Shor--Hastings construction yields non-additive Holevo capacity while the corresponding Holevo barycenters still tensorize, showing that non-additivity of the capacity alone does not force non-multiplicativity of the barycenter.

\section{Direct Sums of Quantum Channels}
For $i=1,\ldots,n$, let $\NN_i:\LL(\HH_i)\to\LL(\KK_i)$ be quantum channels. Set $\HH\equiv\bigoplus_i\HH_i$ and $\KK\equiv\bigoplus_i\KK_i$, and let $P_i$ denote the projection operator in $\HH$ onto $\HH_i$. The direct-sum channel $\NN\equiv\bigoplus_i\NN_i:\LL(\HH)\to\LL(\KK)$ is defined by
\begin{equation*}
    \NN(\rho)\equiv\bigoplus_i\NN_i(P_i\rho P_i).
\end{equation*}
Fukuda and Wolf showed~\cite{Fukuda_2007} that
\begin{equation*}
    \chi(\NN)=\log\sum_i2^{\chi(\NN_i)},
\end{equation*}
and that an optimal ensemble for $\NN$ may be formed as a block-diagonal
convex combination of optimal ensembles for the individual channels, where the ensemble associated with the $i$-th block is assigned the weight
$\lambda_i=2^{\chi(\NN_i)}/\sum_k2^{\chi(\NN_k)}$. Consequently,
\begin{equation}
    \Omega(\NN)
    =\bigoplus_i
    \frac{2^{\chi(\NN_i)}}{\sum_k2^{\chi(\NN_k)}}
    \Omega(\NN_i).
\label{eq: barycenter of dir sum}
\end{equation}
Since $\NN\otimes\NN=\bigoplus_{i,j}\NN_i\otimes\NN_j$, which is again a direct-sum channel, we can apply Eq.~\eqref{eq: barycenter of dir sum} to the product channel $\NN\otimes\NN$, and expanding $\Omega(\NN)\otimes\Omega(\NN)$ yields
\begin{align}
    \Omega(\NN\otimes\NN)
    &=
    \bigoplus_{i,j}
    \frac{2^{\chi(\NN_i\otimes\NN_j)}}
    {\sum_{k,l}2^{\chi(\NN_k\otimes\NN_l)}}
    \Omega(\NN_i\otimes\NN_j),
    \label{eq: bary of tens square}\\
    \Omega(\NN)\otimes\Omega(\NN)
    &=
    \bigoplus_{i,j}
    \frac{2^{\chi(\NN_i)+\chi(\NN_j)}}
    {\sum_{k,l}2^{\chi(\NN_k)+\chi(\NN_l)}}
    \Omega(\NN_i)\otimes\Omega(\NN_j).
    \label{eq: tens square of bary}
\end{align}
In the following two sections, we are going to exploit these formulas by suitable choices of the component channels, to demonstrate our main results.

\section{Non-Multiplicativity of the Holevo Barycenter}
The two states in Eqs.~\eqref{eq: bary of tens square} and~\eqref{eq: tens square of bary} are block diagonal with respect to the same family of mutually orthogonal output subspaces and can therefore coincide only if their corresponding block weights agree. We now construct direct-sum channels for which this necessary condition fails.

\begin{thm}
\label{thm: non-multiplicativity}
There exists a quantum channel $\NN$ such that
\begin{equation*}
    \Omega(\NN\otimes\NN)
    \neq
    \Omega(\NN)\otimes\Omega(\NN).
\end{equation*}
\end{thm}

\begin{proof}
By the Shor--Hastings construction recalled in Sec.~\ref{Sec 2}, we can find channels $\NN_1$ and $\NN_2$ for which the Holevo capacity is strictly superadditive, \ie
\begin{equation}
\label{eq: non-additive chi}
    \chi(\NN_1\otimes\NN_2)=\chi(\NN_1)+\chi(\NN_2)+\delta,
\end{equation}
with $\delta>0$. Let $\NN_0$ be an arbitrary entanglement-breaking channel. Shor proved that~\cite{Shor_2002}
\begin{equation}
\label{eq: add of ent break}
    \chi(\NN_0\otimes\Phi)=\chi(\NN_0)+\chi(\Phi),
\end{equation}
for every quantum channel $\Phi$. Consider the direct-sum channel
\begin{equation}
\label{eq: counter channel}
    \NN=\NN_0\oplus\NN_1\oplus\NN_2.
\end{equation}
For $i,j\in\{0,1,2\}$, let $a_{ij}$ and $b_{ij}$ denote, respectively, the weights of $\Omega(\NN\otimes\NN)$ and $\Omega(\NN)\otimes\Omega(\NN)$ on the $(i,j)$-th block, as given in Eqs.~\eqref{eq: bary of tens square} and~\eqref{eq: tens square of bary}:
\begin{align*}
    a_{ij}
    &=\frac{2^{\chi(\NN_i\otimes\NN_j)}}
    {\sum_{k,l=0}^{2}2^{\chi(\NN_k\otimes\NN_l)}},\\
    b_{ij}
    &=\frac{2^{\chi(\NN_i)+\chi(\NN_j)}}
    {\sum_{k,l=0}^{2}2^{\chi(\NN_k)+\chi(\NN_l)}}.
\end{align*}

We compare $a_{12}/a_{01}$ and $b_{12}/b_{01}$. The normalization factors cancel when taking ratios. Using Eqs.~\eqref{eq: non-additive chi} and~\eqref{eq: add of ent break}, we obtain
\begin{equation*}
    \frac{a_{12}}{a_{01}}
    =2^\delta\frac{b_{12}}{b_{01}}
    >\frac{b_{12}}{b_{01}},
\end{equation*}
where the inequality follows from $\delta>0$. Thus, the two families of block weights cannot coincide, which implies that $\Omega(\NN\otimes\NN)\neq\Omega(\NN)\otimes\Omega(\NN)$
for the channel $\NN$ defined in Eq.~\eqref{eq: counter channel}.
\end{proof}

\section{No Universal Entropic Ordering}
Beyond the non-multiplicativity of the Holevo barycenter, we also show that no universal ordering exists between
$H(\Omega(\NN\otimes\NN))$ and
$H(\Omega(\NN)\otimes\Omega(\NN))$.

\begin{thm}
\label{thm:no-universal-entropy-ordering}
There exist quantum channels $\NN_{-}$ and $\NN_{+}$ such that
\begin{align*}
    H\bigl(\Omega(\NN_{-}\otimes\NN_{-})\bigr)
    &<2H\bigl(\Omega(\NN_{-})\bigr),\\
    H\bigl(\Omega(\NN_{+}\otimes\NN_{+})\bigr)
    &>2H\bigl(\Omega(\NN_{+})\bigr).
\end{align*}
\end{thm}

\begin{proof}
For this purpose, we specialize the preceding direct-sum construction to two component channels, a setting that already suffices to obtain non-multiplicativity. Let $\NN_1$ be a channel obtained through Shor's reduction such that
\begin{equation}
\label{eq: str superadd hol for single N}
    \chi(\NN_1\otimes\NN_1)
    =
    2\chi(\NN_1)+\delta,
\end{equation}
with $\delta>0$. Such a channel $\NN_1$ can be obtained by applying the Shor reduction recalled in Sec.~\ref{Sec 2} to a channel for which the minimum output entropy is strictly subadditive under the tensor product with itself.
The existence of such a channel is guaranteed by Hastings'
result~\cite{Hastings_2009}. Let $d$ denote the output dimension of $\NN_1$. Let $\NN_0$ be a depolarizing channel with the same output dimension, and consider the direct-sum channel
$\NN=\NN_0\oplus\NN_1$.

By the Shor–Hastings construction recalled in Sec.~\ref{Sec 2}, both
$\Omega(\NN_1)$ and $\Omega(\NN_1\otimes\NN_1)$ are maximally mixed on their respective output spaces. Similarly, the symmetry of the depolarizing channel implies that $\Omega(\NN_0)$ is maximally mixed. Moreover, King's additivity result for the depolarizing channel gives~\cite{King:depolarising}
\begin{equation*}
    \chi(\NN_0\otimes\Phi)
    =
    \chi(\NN_0)+\chi(\Phi),
\end{equation*}
for every quantum channel $\Phi$. Hence, by the implication from additivity to multiplicativity established in Sec.~\ref{Sec 2},
$\Omega(\NN_i\otimes\NN_j)=\one/d^2$ for all
$i,j\in\{0,1\}$.

For $i,j\in\{0,1\}$, let $\lambda_i$ denote the weight of $\Omega(\NN)$ on the $i$-th block, and let $\mu_{ij}$ denote the weight of $\Omega(\NN\otimes\NN)$ on the $(i,j)$-th block, as given in Eqs.~\eqref{eq: barycenter of dir sum} and~\eqref{eq: bary of tens square}. Define the probability distributions
\begin{equation*}
    \lambda=(\lambda_i\lambda_j)_{i,j\in\{0,1\}},
    \qquad
    \mu=(\mu_{ij})_{i,j\in\{0,1\}}.
\end{equation*}
Since all four block states in both decompositions are maximally mixed on a $d^2$-dimensional output space, it follows that
\begin{align}
    H\bigl(\Omega(\NN)\otimes\Omega(\NN)\bigr)
    &=H(\lambda)+2\log d,
    \label{eq: entropy of tens of bary}\\
    H\bigl(\Omega(\NN\otimes\NN)\bigr)
    &=H(\mu)+2\log d,
    \label{eq: entropy of bary of tens}
\end{align}
where $H(\lambda)$ and $H(\mu)$ denote the Shannon entropies of the corresponding distributions. Therefore, the comparison between the two von Neumann entropies reduces to that between $H(\lambda)$ and $H(\mu)$. By varying the depolarizing parameter, the Holevo capacity of $\NN_0$ can be chosen arbitrarily between $0$ and $\log d$~\cite{King:depolarising}. We now show that either strict ordering can be obtained by appropriately choosing the Holevo capacity of $\NN_0$.

First, choose $\chi(\NN_0)=\chi(\NN_1)$, and denote the resulting direct-sum channel by $\NN_{-}$. Then $\lambda$ is uniform on four outcomes, whereas, up to a permutation of its entries,
$\mu=(1,1,1,2^\delta)/(3+2^\delta)$. Since $\delta>0$, the distribution $\mu$ is nonuniform, whereas
$\lambda$ is uniform. Hence, $\mu\succ\lambda$, strictly, and therefore $H(\mu)<H(\lambda)$. By Eqs.~\eqref{eq: entropy of tens of bary} and~\eqref{eq: entropy of bary of tens}, this establishes the first inequality in the proposition.

For the opposite ordering, choose instead
$\chi(\NN_0)=\chi(\NN_1)+\delta$, which is always possible. Indeed,
%$\NN_1\otimes\NN_1=(\NN_1\otimes\mathrm{id}_d)\circ(\mathrm{id}\otimes\NN_1)$, so
the data processing inequality for the Holevo capacity, together with King's additivity result~\cite{King:depolarising}, gives
$\chi(\NN_1\otimes\NN_1)\leq\chi(\NN_1\otimes\mathrm{id}_d)
=\chi(\NN_1)+\log d$, where $\mathrm{id}_d$ denotes the identity channel on a $d$-dimensional system. Combining this inequality with Eq.~\eqref{eq: str superadd hol for single N}, we obtain
$\chi(\NN_1)+\delta\leq\log d$. Hence, the desired value of $\chi(\NN_0)$ can be attained by a suitable choice of the depolarizing parameter. In non-increasing order, the two distributions $\lambda$ and $\mu$ are given by
\begin{align*}
    \lambda &=\frac{1}{(2^\delta+1)^2}\bigl(2^{2\delta},2^{\delta},2^{\delta},1\bigr),\\
    \mu &=\frac{1}{2^{\delta}+3}\bigl(2^{\delta},1,1,1\bigr).
\end{align*}
A direct comparison shows that
$\lambda\succ\mu$, strictly. This implies that
$H(\lambda)<H(\mu)$, which corresponds to the second inequality in the proposition.
\end{proof}

Actually, looking at the proof, either strict ordering between $\lambda$ and $\mu$ in the sense of majorization can be enforced. Since all the block states are maximally mixed on spaces of the same dimension, these majorization relations lift to the corresponding quantum states. Consequently, neither majorization direction between $\Omega(\MM\otimes\NN)$ and $\Omega(\MM)\otimes\Omega(\NN)$ holds universally for quantum channels $\MM$ and $\NN$.

\section{Discussion}
\label{sec:discussion}
We have shown that the Holevo barycenter of a channel is in general not tensor-multiplicative, which represents a strengthening of celebrated non-additivity of the Holevo capacity. Indeed, our construction, using a direct sum of additive and non-additive channels, shows that there cannot be any simple relation between the tensor product of two Holevo barycenters and the barycenter of the tensor product channel.

This marks a genuine structural departure from previously held intuitions about the geometric stability of output ensembles. Tensoring channels can generate new extremal points in the output ensemble geometry -- points that do not arise from either channel individually --, demonstrating that barycentric structure is sensitive to correlations created by tensor products. This geometric emergence suggests that barycentric behavior cannot be understood solely in single-shot terms: tensor instability must be treated as an intrinsic property in any meaningful classification of quantum channels. In particular, families of channels may be constructed whose tensor products separate entropic or geometric quantities that coincide in isolation, providing new tools for probing the fine structure of quantum information measures.

Moreover, the absence of a universal entropy ordering for barycenters reinforces the idea that barycenter based resource measures are inherently contextual. Their comparison depends on the tensor environment, indicating that barycentric quantities behave more like relational resources than absolute ones. The possibility that barycenters may activate under tensoring—mirroring activation phenomena known for capacities \cite{SY2008,CCH2011}—opens a further line of inquiry: whether barycentric activation can be harnessed to amplify or suppress specific quantum informational properties.

These observations sharpen a longstanding challenge in quantum information theory: the search for a manageable counterexample to the additivity of the Holevo capacity. The geometric mechanisms uncovered here may offer new avenues towards such a construction. Finally, our analysis suggests that the relation between $\chi$-additivity and barycenter multiplicativity may extend beyond finite dimension, raising the prospect of exploring these phenomena in continuous variable channels, where geometric and entropic effects often manifest in amplified form.

\section*{Acknowledgments}
AW thanks Nilanjana Datta for her interest in this question and for initial discussions. 
No LLMs were involved in finding and proving the present results; furthermore, if the writing is atrocious it is simply because none of us is a native speaker of English.
SC and AW are supported by the Alexander von Humboldt Foundation. 
AW acknowledges furthermore support by the German Research Foundation (DFG) under Germany’s Excellence Strategy: Cluster of Excellence ``Matter and Light for Quantum Computing`` (ML4Q) EXC 2004/2–390534769; and by the Spanish MICIN (project PID2022-141283NB-I00) with the support of FEDER funds.

%\newpage

\bibliography{biblio}

\end{document}